\documentclass{llncs}

\usepackage{amsmath, amssymb}

\DeclareMathOperator{\rad}{rad}
\DeclareMathOperator{\rk}{rk}
\DeclareMathOperator{\polylog}{polylog}
\DeclareMathOperator{\chr}{char}

\newcommand{\lL}{{\mathsf{L}}}
\newcommand{\rR}{{\mathsf{R}}}

\newcommand{\sequence}[2]{{{#1}_1,\,\dotsc,\,{#1}_{#2}}}

\renewcommand{\le}{\leqslant}
\renewcommand{\ge}{\geqslant}
\newcommand{\les}{\lessapprox}

\renewcommand{\phi}{\varphi}
\newcommand{\eps}{\varepsilon}
\newcommand{\QQ}{\mathbb{Q}}
\newcommand{\RR}{\mathbb{R}}

\newcommand{\ZZ}{\mathbb{Z}}
\newcommand{\EE}{\mathbb{E}}
\newcommand{\DD}{\mathbb{D}}
\newcommand{\PP}{\mathbb{P}}
\newcommand{\GG}{\mathcal{G}}

\newcommand{\braces}[1]{\left\{{#1}\right\}}
\newcommand{\parens}[1]{\left({#1}\right)}

\newcommand{\abs}[1]{\left\lvert{#1}\right\rvert}

\begin{document}

\title{Bounds for Bilinear Complexity of Noncommutative Group Algebras}
\author{Alexey Pospelov}
\institute{Computer Science Faculty, Saarland University,\\
\email{pospelov@cs.uni-saarland.de}}
\maketitle

\abstract{We study the complexity of multiplication in noncommutative group algebras which is closely related to the complexity of matrix multiplication. We characterize such semisimple group algebras of the minimal bilinear complexity and show nontrivial lower bounds for the rest of the group algebras. These lower bounds are built on the top of Bl\"aser's results for semisimple algebras and algebras with large radical and the lower bound for arbitrary associative algebras due to Alder and Strassen. We also show subquadratic upper bounds for all group algebras turning into ``almost linear'' provided the exponent of matrix multiplication equals $2$.}

\section{Introduction}\label{sec:intro}

We study noncommutative group algebras and the problem of computing the product of two elements of an algebra. We restrict ourselves on the so-called rank or bilinear complexity of multiplication, which, roughly speaking, counts only the bilinear multiplications used by an algorithm, i.e. multiplications where each of the operands depends on one of the input vectors. A quadratic (in terms of dimension of an algebra) upper bound is straightforward, while all currently known general lower bounds are linear.

This research is motivated by the recent group-theoretic approach for matrix multiplication by Cohn and Umans~\cite{Coum} and following group-theoretic algorithms for matrix multiplication~\cite{Cksu}. It was shown that finite groups possessing some special properties can be used to design effective matrix multiplication algorithms. Our goal is to explore the structure of group algebras and investigate structural and complexity relation between noncommutative group algebras and the matrix algebra. We investigate this approach and put it into a different light. In fact, we show that the group algebras for the most promising groups for the group-theoretic approach have essentially the same complexity as the matrix multiplication itself. On the other hand, for a wide class of group algebras a lower bound holds which depends on the exponent of matrix multiplication (denoted in literature by $\omega$, see Sect.~\ref{sec:model} for definition). If one finds a more effective algorithm of multiplication in these group algebras, it would give a better upper bound for $\omega$ (but without necessary proving $\omega=2$, which is the general conjecture~\cite{Burg}). We also study general bilinear complexity of noncommutative group algebras and this paper extends the research in~\cite{Po05,Po08,Chok} where the problem for commutative group algebras over arbitrary fields was solved entirely. Our results also improve the Atkinson's upper bound for the total complexity of multiplication in group algebras~\cite{Atkn}.

Using Bl\"aser's theorem on classification of all algebras of the minimal rank (see Sect.~\ref{sec:bounds}) we formulate a criterion for a semisimple group algebra to be an algebra of the minimal bilinear complexity. For some special cases we also show a $\frac{5}{2}\cdot$dimension-lower bounds for the rank of group algebras. For other special cases we show an up to $3\cdot$dimension of an algebra lower bound. For one special class of groups having not ``too many'' different irreducible representations we show a lower bound which depends on the exponent of matrix multiplications and turns to be superlinear if the exponent of matrix multiplication does not equal to $2$. This employs Sch\"onhage's $\tau$-theorem (see Sect.~\ref{sec:bounds}). We show that this class is not empty, for instance group algebras of symmetric groups of order $n!$ and general linear groups over finite fields have such a lower bound.

Another motivation for this work was the search for algebras of high bilinear complexity. It is known, that over algebraically closed fields there exist families of algebras of arbitrarily high dimensions with bilinear complexity of each algebra from the family strictly greater than $\frac{(\text{dimension of the algebra})^2}{27}$~\cite[Exercise~17.20]{Burg}. However, no concrete examples are known. This is in some sense similar to the situation in logical synthesis theory, where it is known that the circuit complexity (in a full basis) of almost all boolean functions of $n$ variables is asymptotically $c\frac{2^n}{n}$~\cite{Lupa} where the constant $c$ depends solely on the basis, e.g. for the classical circuit basis $\{\vee,\,\&,\,\neg\}$, $c=1$.\footnote{In fact, for a full circuit basis $B=\{f_1,\,\dotsc,\,f_n\}$ where each $f_\nu$ is of $m_\nu$ variables (with no fictitious dependenies) and has weight $w_\nu$, the constant $c=\min_{\substack{1\le\nu\le n\\m_\nu\ge 2}}\frac{w_\nu}{m_\nu-1}$.} But there is no explicit construction of a function of $n$ variables with a superlinear lower bound on the number of gates in a full finite functional basis. We show that a broad class of group algebras has superlinear bilinear complexity if the exponent of matrix multiplication does not equal to $2$.

We then turn to upper bounds and show by a simple technique a general upper bound for the \emph{total} complexity of multiplication in group algebras that depends on the total complexity of matrix multiplication. In fact, if the exponent of matrix multiplication equals $2$, then the total complexity of the multiplication in group algebras is always ``almost linear''. We indicate some special cases, when this upper bound can be improved provided a maximal irreducible representation of the group has not too high dimension.

For lower bounds we distinguish between the \emph{semisimple} and the \emph{modular} case. If the characteristic of the ground field is either zero or does not divide the order of the group then the group algebra is known to be semisimple. In the other case, if the characteristic $p$  divides the order of the group, then the algebra has nontrivial radical. In some cases its structure inside the group algebra can be described exactly. But in general this introduces additional significant difficulties. If the radical has relatively small nilpotence index then it is possible to obtain relatively high lower bounds for the bilinear complexity of multiplication in group algebra.

Finally, we show direct relations between complexity of noncommutative group algebras and complexity of matrix multiplication and pose several open questions.

The paper is organized as follows: in Sect.~\ref{sec:basics} we bring all necessary definitions and notions from algebra and representation theory. In Sect.~\ref{sec:model} we introduce the model of computation we will be working with and formulate related computational problems. We discuss briefly tight relation between different algebraic notions and computational complexity. We introduce an important quantitative measure estimate for complexity of multiplication in families of algebras of growing dimensions which generalizes the well-known notion of the exponent of matrix multiplication. Classical structural results from the theory of finite-dimensional algebras and representation theory will be presented in Sect.~\ref{sec:structure}. Section~\ref{sec:bounds} contains all necessary results from the algebraic complexity theory to be employed for obtaining lower and upper bounds for the complexity of multiplication in group algebras. In Sect.~\ref{sec:lower} we prove the first part of our main result. We show, that for any ``complicated enough'' group its corresponding group algebra is not of the minimal rank. We also prove two different kinds of lower bounds for families of group algebras depending on the representations of their groups. We also show the general relation between the lower bound for the complexity of group algebra multiplication and the complexity of matrix multiplication. We show, that the bilinear complexity of multiplication in group algebras of symmetric groups is superlinear in their dimension if the exponent of matrix multiplication does not equal $2$. In Sect.~\ref{sec:upper} we turn to effective algorithms for multiplication in group algebras. We show the general upper bound for multiplication in any group algebra depending on the exponent of matrix multiplication and some improvements based on particular properties of the group.

\section{Basic Definitions}\label{sec:basics}

In what follows we always use the term \emph{algebra} for an associative algebra with unity. For example, $n\times n$-matrices over some field form an algebra, and so do univariate polynomials over some field modulo some fixed polynomial or multivariate polynomials modulo some system of polynomials.

A \emph{basis} of an algebra is any basis of the underlying vector space. The \emph{dimension} ($\dim A$) of an algebra $A$ is the dimension of the underlying vector space. The multiplication in an algebra is completely defined if it is defined for the vectors of any of its bases: let $A$ be an algebra over $k$, $n=\dim A$, and $\sequence{e}{n}$ be some basis of $A$, then
$$
e_i\cdot e_j=\sum_{\nu=1}^n\alpha_{ij}^\nu e_\nu,\; 1\le i,\,j\le n,
$$
where $\alpha_{ij}^\nu$ are the \emph{structural constants} from the field $k$. We call a basis $\braces{e_i}_{i=1}^n$ of $A$ a \emph{group basis} if the vectors $e_i$ form a multiplicative group with respect to the multiplication in algebra. In this case $A$ is called a \emph{group algebra}. On the other hand, given a finite group $G=\braces{\sequence{g}{n}}$ and a field $k$ we can define a group algebra $k[G]$ as a $n$-dimensional vector space over $k$ with basis $\braces{g_i}_{i=1}^n$ and multiplication in $k[G]$ defined as
$$
\parens{\sum_{i=1}^n \alpha_i g_i}\cdot\parens{\sum_{j=1}^n \beta_j g_j}
=\sum_{\substack{\ell=1\\ g_i g_j=g_\ell}}^n \alpha_i \beta_j g_\ell.
$$

We call the \emph{direct product} of the algebras $A$ and $B$ over one and the same field $k$ the algebra $A\times B$ over $k$ which consists of pairs of vectors $(a,\,b),\,a\in A,\,b\in B$ and all operations in $A\times B$ are performed component-wise: $(a_1,\,b_1)\circ (a_2,\,b_2)=(a_1\circ a_2,\,b_1\circ b_2)$, $\circ\in\{+,\,-,\,\cdot\}$ and $\lambda\cdot(a,\,b)=(\lambda a,\,\lambda b)$, where $a_i\in A,\,b_i\in B,\,i=1,\,2,\,\lambda\in k$.

We call $B\subseteq A$ a \emph{subalgebra} of $A$, if $B$ is a linear subspace of $A$ and the product (in $A$) of any two vectors of $B$ lies in $B$. A subalgebra $I$ of $A$ is called \emph{left \textup{(}right\textup{)} ideal} of $A$ if for all $a\in A,\;x\in I$ the product $ax\in I$ ($xa\in I$ resp.) A left ideal that is at the same time a right ideal is called a \emph{two-sided ideal}. A (left, right, two-sided) ideal is called \emph{maximal} if it is not contained in any other proper (left, right, two-sided) ideal of the algebra. An ideal $I$ is called \emph{nilpotent} if $I^m=\{0\}$ for some $m>0$.\footnote{For a set $S$ with multiplication and a positive integer $r$ $S^r$ denotes the set of all possible products of $r$ elements of $S$: $\{s_1\dotsm s_r:\:s_\rho\in S,\;1\le \rho\le r\}$.} The smallest $m$ with this property is called the \emph{nilpotence index} of $I$. The sum of all nilpotent left ideals of an algebra $A$ is called the \emph{radical} of $A$ and is denoted by $\rad A$. The intersection of all the maximal left ideals of the algebra $A$ is called the \emph{Jacobson radical} of $A$ and is denoted by $J(A)$.

\begin{proposition}\label{prp:rad}
Let $A$ be an algebra over field $k$. Then $\rad A=J(A)$.
\end{proposition}

\begin{proof}
This follows from the fact, that the \emph{descending chain condition} for left ideals in $A$ implies $\rad A=J(A)$, see~\cite{Ward}. It ensures that any family of left ideals in $A$ contains at least one minimal ideal, i.e. an ideal that does not contain any other ideal of the family. In a finite-dimensional algebra this always holds since we can map any family of ideals to the subset of integers in $[0,\,\dim A]$ mapping each ideal to its dimension as a linear subspace. The resulting image will contain the minimal element which will correspond to the set of ideals from the family having the minimal dimension. Obviously, any of these is minimal.
\qed
\end{proof}

The nilpotence index of $\rad A$ will be denoted by $N(A)$. The set of all $x\in\rad A$ such that $x\cdot\rad A=\{0\}$ is called the \emph{left annihilator} of $\rad A$ and is denoted by $\lL_A$. The \emph{right annihilator} $\rR_A$ is introduced in the similar manner.

Algebra $A$ is called a \emph{division algebra} if every element of $A$ has an inverse in $A$ with respect to the multiplication in $A$. $A$ is called \emph{local} if $A/\rad A$ is a division algebra, and $A$ is called \emph{basic} if $A/\rad A$ is a direct product of division algebras. Following Bl\"aser~\cite{Bl04} we call $A$ superbasic if ${A/\rad A\cong k^t}$ for some $t\ge 1$.

Algebra $A$ is called \emph{semisimple} if $\rad A=0$ and \emph{simple} if it does not contain any proper twosided ideals except for the $\{0\}$. Structure of semisimple and simple algebras is described in Wedderburn's theorem which can be found in~\cite{Ward}.

\begin{theorem}
Every finite dimensional semisimple algebra over some field $k$ is isomorphic to a finite direct product of simple algebras. Every finite dimensional simple $k$-algebra $A$ is isomorphic to an algebra $D^{n\times n}$ for an integer $n\ge 1$ and a $k$-division algebra $D$. The integer $n$ and the algebra $D$ are uniquely determined by $A$ \textup(the latter up to isomorphism\textup).
\end{theorem}

\section{Computational Model}\label{sec:model}


Let $U,\,V$, and $W$ be finite dimensional vector spaces over a field $k$. Let ${\phi:\:U\times V\rightarrow W}$ be a bilinear map. A \emph{bilinear algorithm} for $\phi$ is a sequence
$$
(u_1,\,v_1,\,w_1;\;\dotsc;\;u_r,\,v_r,\,w_r)
$$
where $u_\rho\in U^\ast,\,v_\rho\in V^\ast,\,w_\rho\in W$ such that for all $x\in U,\,y\in V$
$$
\phi(x,\,y)=\sum_{\rho=1}^r u_\rho(x)v_\rho(y)w_\rho.
$$
$r$ is called the \emph{length} of the bilinear algorithm and the minimal length over all bilinear algorithms for $\phi$ is called the \emph{rank} or the \emph{bilinear complexity} of $\phi$ and is denoted by $\rk\phi$.

A sequence
$$
(u_1,\,v_1,\,w_1,\;\dotsc,\;u_\ell,\,v_\ell,\,w_\ell)
$$
where $u_\lambda,\,v_\lambda\in(U\times V)^\ast,\,w_\lambda\in W$ such that for all $x\in U,\,y\in V$
$$
\phi(x,\,y)=\sum_{\lambda=1}^\ell u_\lambda(x,\,y)v_\lambda(x,\,y)w_\lambda
$$
is called a \emph{quadratic} algorithm for $\phi$. $\ell$ is called the \emph{length} of the quadratic algorithm and the minimal length over all quadratic algorithms for $\phi$ is called the \emph{multiplicative complexity} of $\phi$ and is denoted by $C(\phi)$. Obviously $C(\phi)\le\rk\phi$. A straightforward argument implies also that ${\rk\phi\le 2C(\phi)}$ and except for trivial cases, $\rk\phi<2C(\phi)$~\cite{Jaja}.


Multiplication in algebra $A$ is a bilinear map. Rank and multiplicative complexity of multiplication in $A$ are called \emph{rank} and \emph{multiplicative complexity} of $A$ and are denoted by $\rk A$ and $C(A)$ respectively.

Obviously, $\rk A\times B\le\rk A+\rk B$ (also $C(A\times B)\le C(A)+C(B)$). However, it is not known if the converse also holds which is known as the famous Strassen's Direct Sum Conjecture~\cite[p.~360]{Burg}.

Obviously, rank (and therefore, multiplicative complexity) of any algebra $A$ is at most $(\dim A)^2$.

Let $A=\{A_1,\,A_2,\,\dotsc\}$ be a family of algebras over a field $k$. We define $\omega_{A}$, the \emph{rank-exponent of multiplication} in $A$ as
$$
\omega_{A}=\inf\{\tau:\:\rk A_n=O((\dim A_n)^\tau)\text{ for all $n\ge 1$}\}.
$$
Obviously, $0\le\omega_{A}\le 2$. Note that this definition makes only sense if $A$ contains algebras of arbitrarily big dimensions. In this case $\omega_{A} \ge 1$ since multiplication in algebra is always faithful. This notion is very similar to the well-known \emph{exponent of matrix multiplication} which will be denoted just by $\omega$ when the ground field will be clear. The only technical difference is that the exponent of matrix multiplication is defined relative to the square root of the respective algebra dimension. In fact, it can be easily seen that the regular exponent of matrix multiplication equals double the rank-exponent of matrix multiplication.

We acknowledge that the introduced rank-exponent provides quite a crude estimate, since it even does not indicate the growth order of the bilinear complexity as a function of algebra dimension. For example, if $\rk A_n=O(\dim A_n)$, then $\omega_{A}=1$, but the opposite statement must not hold: if $\omega_A=1$ then the rank may potentially be superlinear, e.g. $(\dim A_n)\cdot\polylog(\dim A_n)$. On the other hand, there are no known general upper bounds that are tight enough for the rank-exponent to be too rough. One of the most famous open problems in computational linear algebra and algebraic complexity theory is matrix multiplication, for which its exponent (and the rank exponent) is only known to be within $2\le\omega\le 2.376$~\cite{Cowi}.

\section{Structure of Group Algebras}\label{sec:structure}

Here we introduce some basic concepts from the representation theory. For the extensive treatment we refer to~\cite{Wien}.

Let $G$ be a finite group and $k$ be a field. Then $k[G]$ is semisimple if and only if $\chr k\nmid\sharp G$.

Let $G$ be a finite group and $k$ be an algebraically closed field either of characteristic $0$ or $p\nmid\sharp G$. Then $k[G]$ decomposes into a direct product of matrix algebras:
\begin{equation}\label{eq:dec1}
k[G]\cong k^{n_1\times n_1} \times\dotsm\times k^{n_t\times n_t},
\end{equation}
where each matrix algebra is called \emph{irreducible representation} of $G$ over $k$, and
$$
\sum_{\tau=1}^t n_\tau^2 = \sharp G.
$$
The numbers $n_1,\,\dotsc,\,n_t$ are called the \emph{character degrees} of $G$ in $k$.

If $k$ is not algebraically closed but again of characteristic either $0$ or $p\nmid\sharp G$, then
\begin{equation}\label{eq:dec2}
k[G]\cong D_1^{n_1\times n_1} \times\dotsm\times D_t^{n_t\times n_t},
\end{equation}
where $D_\tau$ are all division algebras over $k$ of dimensions $d_\tau$ for $1\le\tau\le t$ and
$$
\sum_{\tau=1}^t n_\tau^2 d_\tau = \sharp G.
$$

Let $k$ be a field of characteristic $p$ and let $G$ be a finite group of order $np^s$, $p\nmid n$. Suppose that a Sylow $p$-subgroup $P\subseteq G$ is normal. Then $J(k[G])$ is generated by $J(k[P])$ (under the natural inclusion $k[P]\subseteq k[G]$) and
$$
\dim J(k[G])=n(p^s-1).
$$
According to the proposition~\ref{prp:rad}, $J(k[G])=\rad k[G]$ and $k[G]/\rad k[G]$ is semisimple (see~\cite{Ward}). This implies
\begin{equation}\label{eq:dec3}
k[G]/J(k[G])\cong D_1^{n_1\times n_1}\times\dotsm\times D_t^{n_t\times n_t},
\end{equation}
where $D_\tau$ again are all division algebras over $k$ of dimension $d_\tau$ for ${1\le\tau\le t}$ and
\begin{equation}\label{eq:sum3}
\sum_{\tau=1}^t n_\tau^2 d_\tau + \dim J(k[G])=\sharp G.
\end{equation}

In case when Sylow $p$-subgroups of $G$ are not normal the situation becomes more obscure. However, it is known that $J(k[G])$ contains all ideals generated by $J(k[H])$ where $H$ is any normal $p$-subgroup of $G$. In particular, this holds when $H$ is the intersection of all the $p$-Sylow subgroups of $G$.

\section{Bounds for the Rank of Associative Algebras and Complexity of Matrix Multiplication}\label{sec:bounds}

One general lower bound for the multiplicative (and therefore the bilinear) complexity of associative algebras is due to Alder and Strassen.

\begin{theorem}[\cite{Alst}]\label{thm:Alst}
Let $A$ and $B$ be associative algebras over a field $k$ and let $t(A)$ be the number of maximal twosided ideals of $A$. Then
\begin{align}
C(A\times B)&\ge 2\dim A-t(A)+C(B),\label{eq:alst}
\end{align}
\end{theorem}

Algebras for which the Alder-Strassen bound is tight (put $B=\{0\}$ in~\eqref{eq:alst}) are called \emph{algebras of minimal rank}. All such algebras over arbitrary fields were characterized by~Bl\"aser.

\begin{theorem}[\cite{Bl04}]\label{thm:blas}
An algebra $A$ over an arbitrary field $k$ is an algebra of minimal rank iff
\begin{equation}\label{eq:blas}
A\cong C_1\times\dotsm\times C_s\times\underbrace{k^{2\times 2}\times\dotsm\times k^{2\times 2}}_{u\text{ times}}\times B,
\end{equation}
where $C_1,\,\dotsc,\,C_s$ are local algebras of minimal rank with $$\dim (C_\sigma/\rad C_\sigma)\ge 2,$$ i.e., $C_\sigma\cong k[X]/(p_\sigma(X)^{d_\sigma})$ for some irreducible polynomial $p_\sigma(X)$ with ${\deg p_\sigma\ge 2}$, ${d_\sigma\ge 1}$, and $\sharp k\ge 2\dim C_\sigma-2$ and $B$ is a superbasic algebra of minimal rank; that is, there exist $w_1,\dotsc,\,w_m\in\rad B$ with $w_i^2\neq 0$ and $w_i w_j=0$ for $i\neq j$ such that
$$
\rad B=\lL_B+Bw_1 B+\dotsb+Bw_m B=\rR_B+Bw_1 B+\dotsb+Bw_m B
$$
and $\sharp k\ge 2N(B)-2$. Any of the integers $s,\,u$, or $m$ may be zero, and the factor $B$ in \eqref{eq:blas} is optional.
\end{theorem}

The next two lower bounds are due to Bl\"aser.

\begin{theorem}[\cite{Bl00}]\label{thm:bl52}
Let $A$ be a finite dimensional algebra over a field $k$, $A/\rad A\cong A_1\times\dotsm\times A_t$ with $A_\tau=D_\tau^{n_\tau\times n_\tau}$ for all $\tau$, where $D_\tau$ is a $k$-division algebra. Assume that each factor $A_\tau$ is noncommutative, that is, $n_\tau\ge 2$ or $D_\tau$ is noncommutative. Let $n=n_1+\dotsb+n_t$. Then
$$
\rk A\ge\frac{5}{2}\dim A-3n.
$$
\end{theorem}

We will show later how this can be combined with Theorem~\ref{thm:Alst} for group algebras to obtain high lower bounds in cases when some $A_\tau$ are commutative. The next theorem gives a particularly good lower bound for algebras with big radical and small nilpotence index.

\begin{theorem}[\cite{Bl00}]
Let $k$ be a field and $A$ be a finite dimensional $k$-algebra. For all $m,\,n\ge 1$, the rank of $A$ is bounded by
\begin{multline}\label{eq:blrad}
\rk A\ge\dim A-\dim((\rad A)^{n+m-1})\\
+\dim((\rad A)^m)+\dim((\rad A)^n).
\end{multline}
\end{theorem}

The following fact is a simplified version of Sch\"onhage's $\tau$-theorem.

\begin{theorem}[\cite{Schn}]
Let
$$
A=k^{n_1\times n_1}\times\dotsm\times k^{n_t\times n_t},
$$
where $n_\tau>1$ for at least one $\tau$ and $\rk A\le r$. Let $\omega_0$ be a root of the equation
$$
n_1^{x}+\dotsb+n_t^{x}=r.
$$
Then the exponent of matrix multiplication over $k$ does not exceed $\omega_0$.
\end{theorem}

\section{Lower Bounds}\label{sec:lower}

Let $\GG=\{G_1,\,G_2,\dotsc\}$ be a family of finite groups of unbounded orders and let $k$ be a field. We will distinguish between two different cases:
\begin{enumerate}
\item\label{item:semisimple} $\chr k=0$ or $\chr k=p$ and for any $i\ge 1$ $p\nmid\sharp G_i$ and
\item\label{item:modular} $\chr k=p$ and for some $i\ge 1$ $p\mid\sharp G_i$.
\end{enumerate}

We will call $\GG$ in the first case a \emph{semisimple} family of groups and in the second a \emph{modular} family of groups. We will start with the semisimple case.

\subsection{Semisimple Case}\label{subsec:lower_semisimple}

We will start with the case of algebraically closed $k$ since all simple algebras over $k$ are simply matrix algebras.

\begin{lemma}\label{lem:ineq}
Let $n_1,\,\dotsc,\,n_t\ge 0$ and $\delta\ge 1$. Then
\begin{equation}\label{eq:ineq2}
\sum_{\tau=1}^t n_\tau\le t^{1-\frac{1}{\delta}}\parens{\sum_{\tau=1}^t n_\tau^\delta}^{\frac{1}{\delta}}.
\end{equation}
\end{lemma}
\begin{proof}
Let $x_1,\,\dotsc,\,x_t,\,y_1,\,\dotsc,\,y_t$ be complex numbers and $a,\,b\ge 1$ be such that ${\frac{1}{a}+\frac{1}{b}=1}$. Then, by H\"older's inequality
$$
\sum_{\tau=1}^t \abs{x_\tau} \abs{y_\tau}\le\parens{\sum_{\tau=1}^t \abs{x_\tau}^a}^\frac{1}{a}\parens{\sum_{\tau=1}^t \abs{y_\tau}^b}^{\frac{1}{b}}.
$$
Choosing $x_\tau=n_\tau$ and $y_\tau=1$ for all $\tau$, $a=\delta$, and $\frac{1}{b}=1-\frac{1}{\delta}$ completes the proof.
\qed
\end{proof}

Let $G$ be a finite group and $k$ be a field. We introduce following notation: let $t_i(G)$ be the number of irreducible character degrees of $G$ over $k$ equal to $i$. Let $T_i(G)=\sum_{j=i}^\infty t_j(G)$ be the number of irreducible character degrees of $G$ over $k$ not less than $i$. Obviously,
\begin{align*}
T_i(G)&\ge T_j(G),\text{ if }i<j;\\
t_i(G)&=T_i(G)-T_{i+1}(G);\\
\sharp G&=\sum_{i=1}^\infty i^2 t_i(G);\\
t_i(G)&=0,\text{ if }i\ge\sqrt{\sharp G-1}.
\end{align*}

The last follows from the fact, that every group has at least two different irreducible representations. Note, that the number of maximal twosided ideals of $k[G]$ is exactly $T_1(G)=t$, where $t$ is the number of multiplicands in~\eqref{eq:dec1}.

\begin{theorem}
Let $G$ be a finite group and $k$ be an algebraically closed field of characteristic either $0$ or $p\nmid\sharp G$. Let $t$ be as in~\eqref{eq:dec1}.
\begin{enumerate}
\item If $T_3(G)=0$ then $k[G]$ is of minimal rank and
$$
\rk k[G]=2\sharp G-t=t_1(G)+7t_2(G).
$$
\item If $T_3(G)>0$ then $k[G]$ is not of minimal rank then
$$
\rk k[G]\ge 2\sharp G-t+\max\parens{\frac{5}{2}T_7(G),\,1}.
$$
\item Let $\GG=\braces{G_1,\,G_2,\,\dotsc}$ be a family of finite groups, ${\sharp G_n<\sharp G_{n+1}}$ for all $n\ge 1$. Assume that the number of irreducible character degrees of $G\in\GG$ over $k$ is $o(\sharp G)$.\footnote{By using this notation we mean that for any constant $c>0$ there exists such $N>0$ that if $G\in\GG$ and $\sharp G>N$ then the number of irreducible character degrees of $G$ over $k$ is smaller than $c\cdot\sharp G$.} Then the following lower bound holds:
$$
\rk k[G]\ge \frac{5}{2}\sharp G-o(\sharp G).
$$
\end{enumerate}
\end{theorem}

\begin{proof}
Consider the decomposition~\eqref{eq:dec1} for $k[G]$. Note, that the number $t$ is exactly the number of maximal twosided ideals of $k[G]$. Assume w.l.o.g. that $n_1\le\dotsb\le n_t$ and let $A$ be the direct product of all the matrix algebras from~\eqref{eq:dec1} of order $1$ or $2$ and let $B$ be the remaining product: $k[G]=A\times B$. Note, that
\begin{align}\label{eq:Asize}
\dim A&=t_1(G)+4t_2(G)=T_1(G)+3T_2(G)-4T_3(G),\\
\rk A&=t_1(G)+7t_2(G)=2\dim A-(t_1(G)+t_2(G)).\label{eq:Arank}
\end{align}
\eqref{eq:Arank} and the fact that $A$ is of minimal rank follow from Theorem~\ref{thm:blas}. The number of maximal twosided ideals in $A$ is $t_1(G)+t_2(G)$.
\begin{enumerate}
\item Let $k[G]=A$. Then $T_3(G)=0$, $t=t_1(G)+t_2(G)$ and theorem follows from~\eqref{eq:Arank}.
\item\label{case:lower2} Let $B$ be nonempty. By Theorem~\ref{thm:blas}, $k[G]$ is not of minimal rank, therefore $\rk k[G]\ge 2\sharp G-t+1$. By~\eqref{eq:alst} and the fact that $A$ is of minimal rank
$$
\rk k[G]=\rk A\times B=2\dim A-(T_1(G)-T_3(G))+\rk B.
$$
The lower bound follows from~\eqref{eq:alst} and the upper from the trivial inequality $\rk A\times B\le\rk A+\rk B$. Let $B=B_1\times B_2$ where $B_1$ contains all matrix algebras of~\eqref{eq:dec1} of order $\le 6$. The number of maximal twosided ideals in $B_1$ is $t_3(G)+\dotsb+t_6(G)=T_3(G)-T_7(G)$. Then, using~\eqref{eq:alst} once again
$$
\rk B\ge 2\dim B_1-(T_3(G)-T_7(G))+\rk B_2.
$$
Assume that $B_2$ is not empty. Recall, that $n_1\le\dotsb\le n_t$ and therefore $n_{t-T_7(G)+1}\ge 7$. For $B_2$ we can use Theorem~\ref{thm:bl52}:
\begin{multline*}
\rk B_2\ge\frac{5}{2}\sum_{\tau=t-T_7(G)+1}^t n_\tau^2-3\sum_{\tau=t-T_7(G)+1}^t n_\tau\\
=2\dim B_2+\sum_{\tau=t-T_7(G)+1}^t\parens{n_\tau\parens{\frac{n_\tau}{2}-3}}\ge 2\dim B_2+\frac{7}{2}T_7(G).
\end{multline*}
Gathering it all together, we get
\begin{multline*}
\rk k[G]\ge 2\dim A+2\dim B_1+2\dim B_2-T_1(G)+\frac{5}{2}T_7(G)\\
=2\sharp G-t+\frac{5}{2}T_7(G),
\end{multline*}
which proves the second statement of the theorem.
\item
Let $t=o\parens{\sharp G}$. Let $k[G]=k^{t_1(G)}\times C$, $C$ is obviously not empty, and $\dim C=n_{t-T_2(G)+1}^2+\dotsb+n_t^2$. By Alder-Strassen theorem
$$
\rk k[G]=\rk k^{t_1(G)}+\rk C\ge t_1(G)+\frac{5}{2}\dim C-3\sum_{\tau=t-T_2(G)+1}^t n_\tau.
$$
By using Lemma~\ref{lem:ineq} for dimensions of factors of $C$ and setting $\delta=\frac{1}{2}$ we obtain
$$
\sum_{\tau=t-T_2(G)+1}^t n_\tau\le\sqrt{T_2(G)\dim C}\le\sqrt{t\sharp G}=o(\sharp G).
$$
On the other hand, the number $t_1(G)$ of different irreducible representations of $G$ of dimension $1$ does not exceed $t$ and therefore is also $o(\sharp G)$, therefore, $\dim C=\sharp G-t_1(G)=\sharp G-o(\sharp G)$. Therefore, $\rk k[G]\ge\frac{5}{2}\sharp G-o(\sharp G)$ \qed
\end{enumerate}
\end{proof}

\begin{remark}
The lower bound in case~\ref{case:lower2} can be improved further by employing the lower bound due to Bl\"aser $\rk k^{n\times n}\ge 2n^2+n-2$ for $n\ge 3$~\cite{Bl03}. However, the best we can achieve by now is to employ Alder-Strassen lower bounds for all multiplicands in~\eqref{eq:dec1} except for one (of the biggest dimension) and use $2n^2+n-2$ for the last: if $n_1\le\dotsb\le n_t$ and $n_t\ge 3$ then
$$
\rk k^{n_1\times n_1}\times\dotsm\times k^{n_t\times n_t}\ge2\sharp G+n_t-t-1.
$$
\end{remark}

\begin{corollary}\label{cor:low}
Let $k$ be an algebraically closed field of characteristic $0$.
\begin{enumerate}
\item Let $S_n$ be the symmetric group of order $n!$. Then
$$
\rk k[S_n]\ge\frac{5}{2}n!-o(n!).
$$
\item Let $GL(2,\,q)$ be the general linear group of nonsingular $2\times 2$-matrices over $GF(q)$. Then
$$
\rk k[GL(2,\,q)]\ge\frac{5}{2}\sharp GL(2,\,q)-o(\sharp GL(2,\,q)).
$$
\item Let $SL(2,\,q)$ be the special linear group of $2\times 2$-matrices over $GF(q)$ with determinant $1$. Then
$$
\rk k[SL(2,\,q)]\ge\frac{5}{2}\sharp SL(2,\,q)-o(\sharp SL(2,\,q)).
$$
\item Let $p_n$ be the $n$th prime number. Let $F_{p_n,\,p_n-1}$ be a Frobenius group of order $p_n(p_n-1)$ defined by $\{a,\,b:a^{p_n}=b^{p_n-1}=1,\,b^{-1}ab=a^u\}$, where $u$ is an element of order $p_n-1$ in $\ZZ^\ast_{p_n}$~\cite{Jali}. Then
$$
\rk k[F_{p_n,\,p_n-1}]\ge\frac{5}{2}p_n^2-o(p_n^2).
$$
\item Let $p_n$ be the $n$th prime number and let $G_n$ be a non-abelian $p_n$-group with an abelian subgroup of index $p_n$. Then
$$
\rk k[G_n]\ge\frac{5}{2}\sharp G-o(\sharp G).
$$
\end{enumerate}
\end{corollary}
\begin{proof}
\begin{enumerate}
\item The statement follows from the fact that the number of different irreducible representations of $S_n$ over $k$ equals the number of partitions of $n$~\cite{Jake} which asymptotically is $\frac{e^{\pi\sqrt{\frac{2n}{3}}}}{4n\sqrt{3}}=o(n!)$~\cite{Haru}, the latter can be observed easily from the well-known asymptotic of factorial: $n!\sim\sqrt{2\pi n}\parens{\frac{n}{e}}^n$.
\item~\cite{Jali} The number of elements in $GL(2,\,q)$ equals $q^4-q^3-q^2+q\ge\frac{3}{8}q^4$, The number of different irreducible representations of $GL(2,\,q)$ is $q^2-1=o(q^4)$.
\item~\cite{Coum} The number of elements in $SL(2,\,q)$ equals $q^3-q\ge\frac{3}{4}q^3$. The number of different irreducible representations of $SL(2,\,q)$ is $q-4$ if $q$ is odd and $q-1$ if $q$ is a power of $2$; both are $o(q^3)$.
\item~\cite{Jali} The number of different irreducible representations of $F_{p_n,\,p_n-1}$ is $p_n=o(p_n^2)$.
\item~\cite{Jali} Let $\sharp G_n=p_n^m$. The number of different irreducible representations of $G$ is $p_n^{m-1}+p_n^{m-2}-p_n^{m-3}=p_n^m\parens{\frac{1}{p_n}+\frac{1}{p_n^2}-\frac{1}{p_n^3}}=o(p_n^m)$. \qed
\end{enumerate}
\end{proof}

Note, that if the Direct Sum Conjecture were true, then from~\eqref{eq:dec1} for the rank of multiplication in the group algebra $k[G]$ for algebraically closed $k$ would immediately follow
$$
\rk k[G]=\rk k^{n_1\times n_1}+\dotsb+\rk k^{n_t\times n_t}.
$$
It turns out that an insignificantly weaker version of the corresponding lower bound can be proved independently of the validity of the Direct Sum Conjecture.

\begin{theorem}
Let $\GG=\braces{G_1,\,G_2,\,\dotsc}$ be a family of finite groups and $k$ be an algebraically closed field whose characteristic does not divide any of the orders of groups from $\GG$. Let $f(N)$ be a function that for each $G\in\GG$ the dimension of the largest irreducible representation of $G$ is at least $f(\sharp G)$. Then
$$
\rk k[G]\ge f(\sharp G)^{\omega},
$$
where $\omega$ is the exponent of matrix multiplication over $k$. Let $t(N)$ be a function such that for each $G\in\GG$ the number of different irreducible representations of $G$ does not exceed $t(\sharp G)$. Then
$$
\rk k[G]\ge\frac{(\sharp G)^{\frac{\omega}{2}}}{t(\sharp G)^{\frac{\omega^2}{4}-\frac{\omega}{2}}}
$$
\end{theorem}

\begin{proof}
The first statement trivially follows from the observation that for any algebras $A,\,B$ over one field $\rk A\times B\ge\max\{\rk A,\,\rk B\}$.

Let $k[G]$ have decomposition according to~\eqref{eq:dec1}. Consider the following equation
$$
n_1^x+\dotsb+n_t^x=\rk k[G].
$$
Let $\omega_0$ be a root of this equation. Then by Sch\"onhage's $\tau$-theorem $\omega\le\omega_0$. In other words, using the fact that all $n_\tau\ge 1$
$$
n_1^\omega+\dotsb+n_t^\omega\le\rk k[G].
$$
On the other hand, by employing Lemma~\ref{lem:ineq}
$$
\rk k[G]\ge\sum_{\tau=1}^t n_\tau^\omega=\sum_{\tau=1}^t (n_\tau^2)^{\frac{\omega}{2}}\ge\parens{t^{1-\frac{\omega}{2}}\cdot\sum_{\tau=1}^t n_\tau^2}^{\frac{\omega}{2}}\ge\frac{(\sharp G)^{\frac{\omega}{2}}}{t(\sharp G)^{\frac{\omega^2}{4}-\frac{\omega}{2}}}.
$$
which proves the theorem.\qed
\end{proof}

\begin{corollary}
\begin{enumerate}
\item If the number of different irreducible representations of groups in the family does not grow ``too fast'' then the exponent of matrix multiplication is at most twice the rank exponent of the corresponding family of group algebras. More precisely, if ${t(N)=o(N^\eps)}$ for any $\eps>0$ then $\omega_{k[G]}\ge \frac{\omega}{2}$.
\item In the same setting, if $\omega > 2$, then the rank of group algebras from the family described above is superlinear on their dimensions.
\item If $\omega>2$ and $f(N)\gg N^{\frac{1}{\omega}}$ then the group algebras from the corresponding family of groups have superlinear bilinear complexity. One promising family of finite groups which could help to achieve $\omega=2$ in~\cite{Coum} has $f(N)=N^{\frac{1}{2}-\eps}$ for some fixed $\eps>0$. It follows, that in general one should look for $\eps>\frac{1}{2}-\frac{1}{\omega}>0.079$ since otherwise the lower bound depends on $\omega$ and is not superlinear iff $\omega=2$.
\item If $t(N)\ll N^{\frac{2}{\omega}}$ then the bilinear complexity of the corresponding group algebras is superlinear provided $\omega>2$. In particular, this holds if $t(N)\le N^{0.841}$.
\end{enumerate}
\end{corollary}

\begin{corollary}
Let $k$ be an algebraically closed field of characteristic $0$.
\begin{enumerate}
\item Let $\{S_n\}_{n\ge 1}$ be the family of symmetric groups, $S_n$ to be of order $n!$. Then $\omega_{k[S_n]}=\frac{\omega}{2}$.
\item Let $\{GL(n,\,q)\}_{n\ge 1}$, $q$ fixed, be the family of general linear groups of nonsingular $n\times n$-matrices over $GF(q)$. Then $\omega_{k[GL(n,\,q)]}=\frac{\omega}{2}$.
\end{enumerate}
\end{corollary}

\begin{proof}
\begin{enumerate}
\item
For the proof refer to Corollary~\ref{cor:low}.
\item
The order of $GL(n,\,q)$ is
$$
N=\prod_{i=1}^{n-1}\parens{q^n-q^i}=q^{n^2}\underbrace{\prod_{i=1}^{n-1}\parens{1-\frac{1}{q^i}}}_{=:Q}.
$$
Note that $\parens{1-\frac{1}{q}}^{n-1}\le Q\le 1$. $GL(n,\,q)$ has an analytical irreducible representation of order
$$
d=\prod_{i=1}^{n-1}\parens{q^i-1}=\prod_{i=1}^{n-1}q^i\parens{1-\frac{1}{q^i}}=q^{\frac{n(n-1)}{2}}Q,
$$
\cite{Gelf}. It follows, that at least one irreducible representation of has the same order. Now the corresponding matrix algebra has dimension
$$
d^2=q^{n^2-n}Q^2=N\frac{Q}{q^n}.
$$
We will show now that $\frac{q^n}{Q}=o(N^\eps)$ for any $\eps>0$. This will complete the proof since
$$
\rk k[GL(n,\,q)]\ge d^{\omega}=\parens{d^2}^{\frac{\omega}{2}}\ge N^{(1-\eps)\frac{\omega}{2}}
$$
for all groups of size $N>N_0$ and $\eps>0$ where $N_0$ depends on the choice of $\eps$.
\begin{gather*}
\frac{q^n}{Q}\le\frac{q^n}{\parens{1-\frac{1}{q}}^{n-1}}\le q^{2n-1}.\\
N^{\eps}\ge q^{\eps n^2}\parens{1-\frac{1}{q}}^{\eps(n-1)}\ge q^{\eps n^2-\eps n}.
\end{gather*}
So $N^\eps>\frac{q^n}{Q}$ if $n>\frac{2}{\eps}+1$.\qed
\end{enumerate}
\end{proof}

\subsection{Modular Case}\label{subsec:lower_modular}

Let $k$ be now an algebraically closed field of characteristic $p$ and let $G$ be a finite group of order $N=np^d$, where $p\nmid n$. We will assume that $G$ has \emph{the} normal Sylow $p$-subgroup $H$ of order $p^d$. In this case $\rad k[G]$ is generated by the \emph{augmentation} ideal\footnote{The augmentation ideal of a group algebra $A$ with a group basis $\{e_1,\,\dotsc,\,e_n\}$ is the ideal generated by all vectors $\sum x_i e_i$ with $\sum x_i=0$.} of $k[H]$ and $\dim\rad k[G]=p^d(n-1)$.

We will further be concerned with the case of abelian $H$, which is then a direct product of cyclic $p$-groups:
\begin{align}
H&=\ZZ_{p^{t_1}}\times\dotsm\times\ZZ_{p^{t_s}},& t_1&\ge\dotsb\ge t_s,& d&=t_1+\dotsb+t_s.\label{eq:sylow}
\end{align}
We will denote elements of $H$ by $h_{i_1,\,\dotsc,\,i_s}$, $0\le i_\sigma<p^{t_\sigma}$ for all $1\le\sigma\le s$ assuming
$$
h_{i_1,\,\dotsc,\,i_s}\cdot h_{j_1,\,\dotsc,\,j_s}=h_{(i_1+j_1)\bmod p^{t_1},\,\dotsc,\,(i_s+j_s)\bmod p^{t_s}}.
$$
Let
\begin{gather*}
r_1=h_{1,\,0,\,0,\,\dotsc,\,0}-h_{0,\,0,\,0,\,\dotsc,\,0},\\
r_2=h_{0,\,1,\,0,\,\dotsc,\,0}-h_{0,\,0,\,0,\,\dotsc,\,0},\\
\dotsc\\
r_s=h_{0,\,0,\,0,\,\dotsc,\,1}-h_{0,\,0,\,0,\,\dotsc,\,0}.
\end{gather*}
The augmentation ideal of $k[H]$ (and $R=\rad k[G]$) is generated by $r_1,\,\dotsc,\,r_s$. It is easy to see that $r_\sigma^{p^{t_\sigma}}=0$ and the system of vectors
$$
\braces{r_1^{i_1}\dotsm r_s^{i_s}\;|\;i_1+\dotsb+i_s\ge 1,\;0\le i_\sigma<p^{t_\sigma}}
$$
is linearly independent. The system
$$
\braces{r_1^{i_1}\dotsm r_s^{t_s}\;|\;i_1+\dotsb+i_s\ge m,\;0\le i_\sigma<p^{t_\sigma}}
$$
is also linearly independent and generates $R^m$, so $\dim R^m=n(p^d-a_{m-1})$ where
$$
a_{m-1}=\sharp\braces{(i_1,\,\dotsc,\,i_s)\;|\;i_1+\dotsb+i_s\le m-1,\;0\le i_\sigma< p^{t_\sigma}}.
$$

Let $\xi$ be a discrete random variable. We denote by $\EE\xi$ the expectation of $\xi$, i.e. if $\xi$ takes value $a_i\in\RR$ with probability $p_i\ge 0$ for $1\le i\le n$, $\sum_{i=1}^n p_i=1$, then $\EE\xi=\sum_{i=1}^n a_i p_i$. We also denote by $\DD\xi=\EE(\xi-\EE\xi)^2$ the dispersion of $\xi$.

\begin{theorem}
Let $\GG=\{G_1,\,G_2,\,\dotsc\}$ be a family of groups and $k$ be a field of characteristic $p$. Let $G\in\GG$ and $\sharp G=N=np^d$, where $p\nmid n$. Assume that $P=Z(G)$\footnote{$Z(G)$ is the center of $G$, i.e. the set of elements of $G$ that commute with all the elements of $G$.} is the Sylow $p$-subgroup of $G$ and the parameter $d$ is unbounded for groups in $\GG$. Let $p^T$ be the order of biggest cyclic factor of $P$ and $p^t$ be the smallest order, and let $s$ be the total number of factors. Assume that for any $\eps>0$ the difference $T-t<\frac{1}{2}\log_p\eps s$ for all $G\in\GG$ with $\sharp G>N_0=N_0(\eps)$. Then
$$
C(k[G])\ge \parens{2+\frac{1}{n}}\sharp G-o(\sharp G).
$$
\end{theorem}
\begin{proof}
Following proof is based on ideas by Chokayev and generalizes similar result proven in~\cite{Chok} for one special case of commutative group algebras.

We note, that since $P$ is abelian, it is a finite product of cyclic $p$-groups:
$$
P=\ZZ_{p^{t_1}}\times\dotsm\times\ZZ_{p^{t_s}}
$$
where $t_1\le\dotsb\le t_s$ and the exponent of $P$ is $p^{t_s}$. Since it is $o(\sharp P)$, the parameter $s$ is unbounded among all groups from $\GG$.

According to~\eqref{eq:blrad}
\begin{multline*}
C(k[G])\ge\sharp G+n(p^d-a_{m-1})+n(p^d-a_{m-1})-n(p^d-a_{2m-1})\\
=\parens{2+\frac{a_{2m-1}-2a_{m-1}}{np^d}}\sharp G.
\end{multline*}
We will show now that we may choose $m$ in such a way that $\frac{a_{2m}}{p^d}\rightarrow 1$, $\frac{a_m}{p^d}\rightarrow 0$ when $s\rightarrow\infty$. Consider indices $\{i_\sigma\}_{\sigma=1}^s$ as independent random variables with $i_\sigma$ taking value in $[0,\,p^{t_\sigma}-1]$ with probability $\frac{1}{p^{t_\sigma}}$ for $1\le\sigma\le s$. Then
\begin{align*}
\EE i_\sigma&=\frac{p^{t_\sigma}-1}{2},&\DD i_\sigma&=\frac{p^{2t_\sigma}-1}{12},
\end{align*}
and denoting $\xi_s=i_1+\dotsb+i_s$
\begin{align*}
\EE\xi_s&=\frac{1}{2}\sum_{\sigma=1}^s p^{t_\sigma}-\frac{s}{2},&
\DD\xi_s&=\frac{1}{12}\sum_{\sigma=1}^s p^{2t_\sigma}-\frac{s}{12},
\end{align*}
while $\xi_s$ takes each value in $[0,\,\sum_{\sigma=1}^s p^{t_\sigma}-s]$ with probability $\frac{a_m-a_{m-1}}{p^d}$. Now let $m=\frac{2}{3}\EE\xi_s$ be a function of $s$. Then by Chebyshov's inequality
\begin{align*}
\frac{a_{m-1}}{p^d}&=\PP(\xi_s\le m-1)\le\PP(\abs{\xi_s-\EE\xi_s}\ge\EE\xi_s-m+1)\\
&
\le\frac{\DD\xi_s}{(\EE\xi_s-m+1)^2}\le\frac{3sp^{2T}}{4s^2p^{2t}}=\frac{3p^{2T-2t}}{4s}\xrightarrow[s\rightarrow\infty]{} 0,\\
\frac{a_{2m-1}}{p^d}&=\PP(\xi_s\le 2m-1)\ge\PP(\abs{\xi_s-\EE\xi_s}\le 2m-1-\EE\xi_s)\\
&
\ge 1-\frac{\DD\xi_s}{(2m-1-\EE\xi_s)^2}\ge 1-\frac{3p^{2T-2t}}{4s}\xrightarrow[s\rightarrow\infty]{} 1
\end{align*}
which proves the theorem.\qed
\end{proof}
\begin{corollary}
For any field $k$ of characteristic $p$ and any family of groups $\{G_1,\,G_2,\,\dotsc\}$ of growing dimensions there exists a constant $N$ such that the generated family of group algebras $\{k[G_1],\,k[G_2],\,\dotsc\}$ does not contain algebras of minimal rank of dimensions greater than $N$ if their Sylow $p$-subgroups coincide with their centers and contain growing number of cyclic factors of close order.
\end{corollary}

\section{Upper Bounds}\label{sec:upper}

As~\eqref{eq:dec1} and \eqref{eq:dec2} indicate, complexity of multiplication in group algebras is closely related to complexity of matrix multiplication. In particular, provided an effective algorithm for multiplication of square matrices, we immediately obtain an effective algorithm for multiplication in group algebras.

\begin{proposition}\label{lem:sumupp}
Let $\sequence{n}{t}>0$ and $alpha\ge 1$. Then
$$
\sum_{\tau=1}^t n_\tau^\alpha\le\parens{\sum_{\tau=1}^t n_\tau}^\alpha.
$$
\end{proposition}
\begin{proof}
The statement follows from the fact that $x^\alpha$ is convex for $x\ge 0$ and $\alpha\ge 1$.
\end{proof}

For any pair of monotonically growing functions $f(n)$ and $g(n)$ we will write $f(n)\les g(n)$ if for every $\delta>1$ $f(n)\le O\parens{(g(n))^{\delta}}$.

Let $G$ be a finite group and $k$ be an algebraically closed field whose characteristic is either $0$ or does not divide $\sharp G$.
Now we are ready to introduce the general upper bound for the rank of $k[G]$.

\begin{theorem}\label{thm:upgen}
Let $G$ be a group and $k$ be an algebraically closed field of characteristic either $0$ or coprime with $\sharp G$. Then
\begin{equation}\label{eq:upgen}
\rk k[G]\les (\sharp G)^{\frac{\omega}{2}},
\end{equation}
where $\omega$ is the exponent of matrix multiplication.
\end{theorem}
\begin{proof}
Consider decomposition~\eqref{eq:dec1} of $k[G]$ into a direct product of matrix algebras. It follows that
$$
\rk k[G]\le\sum_{\tau=1}^t\rk k^{n_\tau\times n_\tau}.
$$
By definition of the exponent of matrix multiplication
$$
\rk k^{n_\tau\times n_\tau}\le L(k^{n_\tau\times n_\tau})\les n_\tau^\omega.
$$
Thus by Proposition~\ref{lem:sumupp}
$$
\rk k[G]\les \sum_{\tau=1}^t n_\tau^\omega=\sum_{\tau=1}^t (n_\tau^2)^{\frac{\omega}{2}}\le \parens{\sum_{\tau=1}^t n_\tau^2}^{\frac{\omega}{2}}=(\sharp G)^{\frac{\omega}{2}}
$$
which completes the proof.
\qed
\end{proof}

\begin{lemma}\label{lem:upp}
Let $\GG=\{G_1,\,G_2,\,\dotsc\}$ be a family of finite groups and $k$ be an algebraically closed field of characteristic either $0$ or coprime with each $\sharp G_i$. Let $f(N)$ be a function which satisfies following property: for every $G\in\GG$ all character degrees of $G$ over $k$ are less or equal than $f(\sharp G)$. Then for any $G\in\GG$
\begin{equation}\label{eq:upp}
\rk k[G]\les \sharp G\cdot\min_{h(N)}\parens{h(\sharp G)^\omega+\frac{f(\sharp G)^\omega}{h(\sharp G)^2}},
\end{equation}
where $\omega$ is the exponent of matrix multiplication and the minimum is taken over all functions $h(N)$ such that at least one irreducible character degree of $G$ is less or equal than $h(\sharp G)$.
\end{lemma}
\begin{proof}
Let $n_1\ge\dotsb\ge n_t$ be the irreducible character degrees of $G$ over $k$. Let $h(N)$ be as defined. Let $j(N)$ be the number of $n_\tau$ greater than $h(N)$. Note that
$$
\sharp G=\sum_{\tau=1}^t n_\tau^2\ge j(\sharp G)h(\sharp G)^2,
$$
thus $j(N)\le\frac{N}{h(N)^2}$. It follows that
$$
\rk k[G]\les \parens{j(\sharp G)f(\sharp G)^\omega+\sum_{\tau=j(\sharp G)+1}^t n_\tau^\omega}\le \sharp G\frac{f(\sharp G)^\omega}{h(\sharp G)^2}+\sharp G h(\sharp G)^\omega.
$$
The last equation holds for any $h(N)$ so it holds also for the one minimizing the right side.
\qed
\end{proof}

\begin{theorem}\label{thm:upf}
Let $\GG=\{G_1,\,G_2,\,\dotsc\}$ be a family of finite groups and $k$ be an algebraically closed field of characteristic either $0$ or coprime with order of each $G_i$. Let $f(N)$ be a function which satisfies following property: for each $G\in\GG$ all character degrees of $G$ over $k$ are less or equal than $f(\sharp G)$. Then for any $G\in\GG$
\begin{equation}\label{eq:upbound}
\rk k[G]\les \sharp Gf(\sharp G)^{\omega-2+\frac{4}{\omega+2}}\le \sharp G f(\sharp G)^{\omega-1},
\end{equation}
where $\omega$ is the exponent of matrix multiplication.
\end{theorem}
\begin{proof}
It is a well-known fact that every group has at least one (trivial) one-dimensional representation. So we can choose for $h(N)$ in Lemma~\ref{lem:upp} any function which is less than $f(N)$. The result of the theorem follows by choosing $h(N)=f(N)^{1-\frac{2}{\omega+2}}$.
\qed
\end{proof}

\begin{corollary}
\begin{enumerate}
\item If $f(N)=O(1)$ then $\rk k[G_i]=O(N)$.
\item If for any $\eps>0$ $f(N)=o(N^\eps)$ then $\omega_{k[G]}=1$.
\end{enumerate}
\end{corollary}

\begin{remark}
\begin{enumerate}
\item Note, that $h(N)=\parens{\frac{2}{\omega}}^{\frac{1}{\omega+2}}f(N)^{1-\frac{\omega}{\omega+2}}$ minimizes the right side of~\eqref{eq:upp}.
\item The upper bound given by~\eqref{eq:upbound} is better than the one given by~\eqref{eq:upgen} if $f(N)=o\parens{N^{\frac{1}{2}-\frac{2}{\omega^2}}}$. According to the best known upper bound $\omega<2.376$ \cite{Cowi}, currently \eqref{eq:upbound} beats \eqref{eq:upgen} if $f(N)=o(N^{0.1457})$.
\end{enumerate}
\end{remark}

Let $k$ now be an arbitrary field of characteristic $0$ and $G$ be a finite group. By definition of prime field, $\QQ\subseteq k$ is the prime subfield of $k$. Let $K\supseteq k$ be an algebraically closed extension of $k$. It is known (see~\cite[Theorem~11.4, Chapter~XVIII]{Lang}) that every representation of $G$ over $K$ is definable over $\QQ(\zeta_m)$ where $m$ is exponent of $G$ and $\zeta_m$ is a primitive $m$-th root of unity. Therefore, it is definable over $k(\zeta_m)$ (if $k$ does not already contain $\zeta_m$). Now consider any irreducible representation of $G$ over $k$. It is a simple $k[G]$-module by Maschke's Theorem \cite[Theorem~1.2, Chapter~XVIII]{Lang}. Therefore, it is isomorphic to $D^{n\times n}$ where $D$ is a $k$-division algebra. $\zeta_m$ is algebraic over $D$ since it is algebraic over $k\subseteq D$ and $D\cong D'\subseteq k(\zeta_m)$. The latter holds since there are no simple irreducible representations of $G$ over $k(\zeta_m)$ other than those isomorphic to matrix algebras over $k(\zeta_m)$.

Thus, $D$ is a subalgebra of $k(\zeta_m)$, or $D\cong k(\zeta_\ell)$ for some $\ell\mid m$.

\begin{theorem}\label{thm:char0}
Let $\GG=\{G_1,\,G_2,\,\dotsc\}$ be a family of finite groups and $k$ be an arbitrary field of characteristic $0$. Then for any $G\in\GG$
$$
\rk k[G]\les (\sharp G)^{\frac{\omega}{2}},
$$
where $\omega$ is again the exponent of matrix multiplication.
\end{theorem}
\begin{proof}
Since $k[G]$ is semisimple, \eqref{eq:dec2} holds. As mentioned above, $D_\tau$ is actually an extension field of $k$, thus for all $\tau$ $\rk D_\tau\le 2d_\tau-1$ since it can be implemented via polynomial multiplication over $k$ and $k$ is infinite. We have
\begin{multline*}
\rk k[G]\les \sum_{\tau=1}^t n_\tau^\omega (2d_\tau-1)<2\sum_{\tau=1}^t n_\tau^\omega d_\tau
=2\sum_{\tau=1}^t \parens{n_\tau^2 d_\tau^{\frac{2}{\omega}}}^{\frac{\omega}{2}}\\
\le 2\parens{\sum_{\tau=1}^t n_\tau^2 d_\tau^{\frac{2}{\omega}}}^{\frac{\omega}{2}}
\le 2\parens{\sum_{\tau=1}^t n_\tau^2 d_\tau}^{\frac{\omega}{2}}=2(\sharp G)^{\frac{\omega}{2}}
\end{multline*}
since $\omega\ge 2$.
\qed
\end{proof}

\begin{remark}
Statement of theorem~\ref{thm:char0} remains true whenever the division algebras appear inside simple irreducible representations of groups have linear rank. Thus,
\begin{enumerate}
\item
Theorem~\ref{thm:char0} holds also when $k$ is finite. It is known that any finite division algebra is an extension field of $k$, by Wedderburn's Little Theorem \cite[Theorem 2.55]{Lidl}, therefore its rank is linear due to Chudnovskys' algorithm, cf. \cite{Chud} or \cite{Shpr}.
\item
It also holds for real closed fields since all division algebras over such fields have bounded dimension (in fact, it can be only 1, 2, 4, or 8)~\cite{Darp}.
\end{enumerate}
\end{remark}

\section{Conclusion}\label{sec:conclusion}

Noncommutative group algebras appear to be closely connected with the matrix algebra. Studying the problem of complexity of multiplication  in group algebras may give us new algebraic insight into this classical problem of computer algebra and algebraic complexity theory. There are numerous open problems related to group algebras. We mention here only some of them.

\begin{enumerate}
\item It could be possible to obtain a general upper bound not depending on the matrix representations for the rank of group algebras based on the group structure that will be better than upper bounds given by Theorems~\ref{thm:upgen},~\ref{thm:upf},~and~\ref{thm:char0}. In this case it could improve the upper bound for matrix multiplication.
\item We would like to extend Theorem~\ref{thm:char0} for fields of arbitrary characteristic that does not divide any of the group orders from the family under consideration.
\item The radical of a group algebra in the modular case is tightly related to Sylow $p$-groups. These groups are well-studied, although their structure may vary very strongly. It is known that the rank of commutative group algebras with nontrivial radical is still linear, so it does not affect the order of the complexity. On the other hand, a commutative group algebra over algebraically closed field of characteristic $p$ is of minimal rank iff its Sylow $p$-group is cyclic. An open question is if similar effects also hold for noncommutative group algebras.
\end{enumerate}

\section*{Acknowledgements}

I would like to thank M.~Bl\"aser for a lot of helpful comments and suggestions and V.~Alekseyev for introducing me into this topic.
This research is supported by Cluster of Excellence ``Multimodal Computing and Interaction'' at Saarland University.

\end{document}